\documentclass[conference,letterpaper]{IEEEtran}
\IEEEoverridecommandlockouts
\usepackage{graphicx}
\usepackage{cite}
\usepackage{epstopdf}
\usepackage{amssymb, amsmath,amsthm, amsfonts}
\usepackage{ifthen}
\usepackage{tikz}
\usepackage{enumerate}
\usepackage{verbatim}
\usepackage{bbm,balance}
\usepackage{amsmath}
\usepackage{url}

\usetikzlibrary{arrows,shapes}
\usetikzlibrary{positioning}
\usetikzlibrary{calc}
\usetikzlibrary{circuits.ee.IEC}

\allowdisplaybreaks

\newtheorem{theorem}{Theorem}

\newtheorem{lemma}{Lemma}

\newtheorem{prop}{Proposition}
\newtheorem{defn}{Definition}
\theoremstyle{definition}
\theoremstyle{definition}\newtheorem{example}{Example}

\newcommand{\be}{\begin{equation}}
\newcommand{\ee}{\end{equation}}
\newcommand{\ben}{\begin{equation*}}
\newcommand{\een}{\end{equation*}}
\newcommand{\ba}{\begin{eqnarray}}
\newcommand{\ea}{\end{eqnarray}}

\allowdisplaybreaks
\begin{document}
\title{Multi-terminal Strong Coordination over Noisy Channels with Encoder Co-operation}
\author{
\IEEEauthorblockN{Viswanathan~Ramachandran}
\IEEEauthorblockA{Department of Electrical Engineering\\
        Indian Institute of Technology Jodhpur, India\\
         vramachandran@iitj.ac.in}\\
        \and
\IEEEauthorblockN{Tobias~J.~Oechtering and Mikael~Skoglund}
\IEEEauthorblockA{Division of Information Science and Engineering\\
        KTH Royal Institute of Technology, Stockholm, Sweden\\
         \{oech,skoglund\}@kth.se }
}
\maketitle

\begin{abstract}
We investigate the problem of strong coordination over a multiple-access channel (MAC) with cribbing encoders. In this configuration, two encoders observe independent and identically distributed (i.i.d.) samples of a source random variable each and encode the inputs to the MAC. The decoder which observes the output of the MAC together with side-information, must generate approximately i.i.d. samples of another random variable which is jointly distributed with the two sources and the side information. We also allow for possible encoder cooperation, where one of the encoders can non-causally crib from the other encoder's input. 
Independent pairwise shared randomness is assumed between each encoder and the decoder at limited rates. Firstly, in the presence of cribbing, we derive an achievable region based on joint source-channel coding. We also prove that in the absence of cribbing, our inner bound is tight for the special case when the MAC is composed of deterministic links, and the sources are conditionally independent given the side information. We then explicitly compute the regions for an example both with and without cribbing between the encoders, and demonstrate that cribbing strictly improves upon the achievable region.
\end{abstract}

\section{Introduction}
The framework of \emph{coordination}~\cite{cuff2010coordination} explores the minimal communication necessary to establish a desired joint distribution of actions among all nodes in a network. In light of the explosion of device-to-device communications as part of the Internet of Things (IoT), this architecture is useful in such scenarios where decentralized cooperation is desired amongst distributed agents. 
We focus on the notion of \emph{strong coordination}, where the distribution of the sequence of actions must be close in total variation to a target distribution. 

\begin{figure}[ht]
\centering
\begin{tikzpicture}[thick, circuit ee IEC]
\node (d1) at (-2,0) [rectangle, draw, right, minimum height=1.0cm, minimum width = 1.2cm]{Enc $1$};
\node (d2) at (2.45,-0.95) [rectangle, draw, right, minimum height=2.5cm, minimum width = 1.6cm]{Dec};
\draw[->] (d2) --++(1.15,0) node[right]{$Y^n$};
\draw[<->] (d1) to[out=60,in=120] node[midway, above] {$K_1 \in [1:2^{nR_{01}}]$} (d2);
\draw[<-] (d1) --++(-1.5,0) node[left]{$X_1^n$};
\node (e2) at (-2,-2) [rectangle, draw, right, minimum height=1.0cm, minimum width = 1.2cm]{Enc $2$};
\draw[<-] (e2) --++(-1.5,0) node[left]{$X_2^n$};
\draw[<-] (d2.80) --++(0,1.5) |- (-2.95,2.1) node[left]{$W^n$};
\node (ch) at (0,-1) [rectangle, draw, right, minimum height=0.8cm, minimum width = 1.2cm]{$p(\tilde{y}|\tilde{x}_1,\tilde{x}_2)$};
\draw[->] (d1) -- (ch) node[midway, above, sloped]{$\tilde{X}_1^n$};
\draw[->] (e2) -- (ch) node[midway, below, sloped]{$\tilde{X}_2^n$};
\draw[->] (ch) -- (d2) node[midway, above]{$\tilde{Y}^n$};
\draw[<->] (e2) to[out=-60,in=-120] node[midway, below] {$K_2 \in [1:2^{nR_{02}}]$} (d2);
\node (j) at (-0.4,-1.7) [rectangle]{};
\draw[->] (j) |- (-0.4,-1) to [make contact={info={},info'={$S$}}] (-1.39,-1) -- (d1.south);
\end{tikzpicture}
\caption{Strong coordination over a MAC with cribbing encoders} \label{fig:encSRnoisy}
\end{figure}
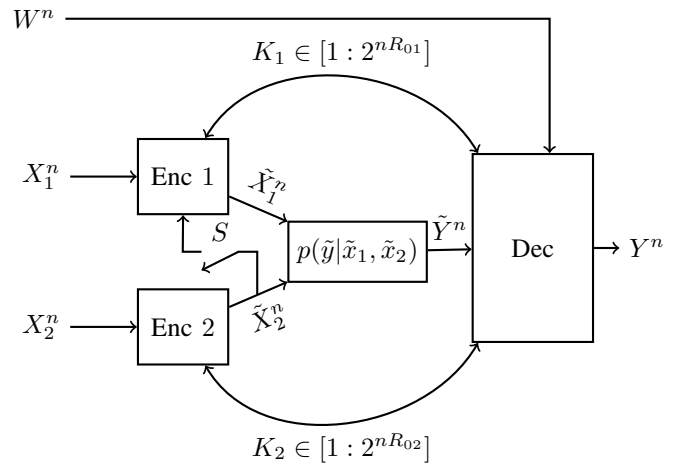

Complete characterizations for strong coordination in multi-terminal networks are comparatively rare. Building upon the the point-to-point network, a cascade network with secrecy constraints was investigated in \cite{SatpathyC16}, for which the optimal trade-off between communication and common randomness rates was determined. 
In \cite{kurri2022multiple,ramachandran2020strong,RamachandranOSIZS2024,RamachandranOSITW2024,ramachandran2024multi}, strong coordination was investigated over a multiple-access network of noiseless links, with a tight characterization derived for independent sources.

Simulation of a channel using another channel as a resource, rather than noiseless communication links as in \cite{cuff2013distributed}, was investigated by \cite{HaddadpourYAG13,HaddadpourYBGAA17}. 
However, multi-terminal extensions of strong coordination over noisy channels with the interplay between different terminals has not received much attention. Accordingly, the current paper extends channel simulation from noisy channels~\cite{HaddadpourYBGAA17} to a three-terminal scenario with possible encoder cooperation~\cite{willems1985discrete,asnani2012multiple,huleihel2016multiple}. 
It also extends the noiseless setting~\cite{kurri2022multiple} to noisy channels. 
The main technical novelty here lies in bringing out the role of encoder cooperation in reducing the shared randomness needed for coordination. 

\noindent \textbf{Main Contributions.}
When the switch $S$ in Fig.~\ref{fig:encSRnoisy} is closed, we derive an achievable region (Theorem~\ref{thm:encsideIBnoisyK0}) based on joint source-channel coding. 
When the switch $S$ in Fig.~\ref{fig:encSRnoisy} is open, we give a tight characterization for the special case when the MAC is composed of deterministic links, and the sources are conditionally independent given the side information (Theorem~\ref{thm:indepnoisy}). The non-trivial part lies in leveraging the deterministic channel and independent sources assumptions to obtain a single-letterization matching the inner bound, which is known to be difficult for distributed source coding settings~\cite{berger1977multiterminal}.
We then explicitly compute the regions for an example both with and without encoder cribbing and demonstrate that cribbing strictly improves upon the achievable region (see Section~\ref{sec:example}).

\section{System Model} \label{UPDATEDsec:SMRnoisy}
The setup comprises two encoders, with Encoder $j \in \{1,2\}$ observing an input given by $X_j^n$, and a decoder which observes a side information sequence $W^n$. For $j \in \{1,2\}$, Encoder $j$ and the decoder can harness pairwise shared randomness $K_j$, assumed to be uniformly distributed on $[1:2^{nR_{0j}}]$.  
When the switch $S$ in Fig.~\ref{fig:encSRnoisy} is closed, Encoder $2$ (which observes $X_2^n$ and has access to $K_2$) first generates the channel input sequence $\tilde{X}_2^n$. Then, Encoder $1$ (which observes $X_1^n$ and has access to $K_1$ as well as $\tilde X_2^n$) creates the channel input sequence $\tilde{X}_1^n$. A discrete-memoryless multiple-access channel specified by $p(\tilde{y}|\tilde{x}_1,\tilde{x}_2)$ maps the channel input sequences into an observation $\tilde{Y}^n$ at the receiver. The sources $(X_{1i},X_{2i},W_i)$, $i=1,2,\ldots,n$, are assumed to be i.i.d. with joint distribution specified by nature as $q_{X_1X_2W}$. The random variables $X_1,X_2,W,\tilde{X}_1,\tilde{X}_2,\tilde{Y}$ assume values in finite alphabets $\mathcal{X}_1,\mathcal{X}_2,\mathcal{W},\mathcal{\tilde{X}}_1,\mathcal{\tilde{X}}_2,\mathcal{\tilde{Y}}$, respectively. The shared randomness indices $K_1$ and $K_2$ are assumed to be independent of each other and of $(X_1^n,X_2^n,W^n)$.
The decoder obtains $(K_1,K_2,W^n,\tilde{Y}^n)$ and simulates an output sequence $Y^n$ (where $Y_i$, $i=1,\dots,n,$ assumes values in a finite alphabet $\mathcal{Y}$) which along with the input sources and side information must be approximately i.i.d. according to the joint distribution $q^{(n)}_{X_{1}X_{2}WY}(x_1^n,x_2^n,w^n,y^n):=\prod_{i=1}^n q_{X_1X_2WY}(x_{1i},x_{2i},w_i,y_i)$ (refer Figure~\ref{fig:encSRnoisy}). 

\begin{defn}\label{defn:codenoisy}
A $(2^{nR_{01}}, 2^{nR_{02}}, n)$ \emph{code} comprises two randomized encoders $p^{\emph{Enc}_2}(\tilde{x}_2^n|x_2^n,k_2)$ and $p^{\emph{Enc}_1}(\tilde{x}_1^n|x_1^n,k_1,\tilde x_2^n)$ and a randomized decoder $p^{\emph{Dec}}(y^n|k_1,k_2,w^n,\tilde{y}^n)$, where $k_j\in[1:2^{nR_{0j}}]$, $j \in \{1,2\}$.
\end{defn}
The induced joint distribution of all the random variables and the resulting induced marginal distribution on $(X_1^n,X_2^n,W^n,Y^n)$ are respectively given by
\begin{align*}
&p(x_1^n,x_2^n,w^n,k_1,k_2,\tilde{x}_1^n,\tilde{x}_2^n,\tilde{y}^n,y^n) \notag\\
&=\frac{1}{2^{n(R_{01}+R_{02})}}q(x_1^n,x_2^n,w^n) p^{\text{Enc}_2}(\tilde{x}_2^n|x_2^n,k_2) \\
& \hspace{6pt}\times p^{\text{Enc}_1}(\tilde{x}_1^n|x_1^n,k_1,\tilde x_2^n) p(\tilde{y}^n|\tilde{x}_1^n,\tilde{x}_2^n) p^{\text{Dec}}(y^n|k_1,k_2,w^n,\tilde{y}^n),
\end{align*}
and
\begin{align*}
&p^{\text{ind}}(x_1^n,x_2^n,w^n,y^n)\\
&\hspace{12pt}=\sum_{k_1,k_2,\tilde{x}_1^n,\tilde{x}_2^n,\tilde{y}^n}p(x_1^n,x_2^n,w^n,k_1,k_2,\tilde{x}_1^n,\tilde{x}_2^n,\tilde{y}^n,y^n).
\end{align*}

\begin{defn} \label{def:achnoisy}
A rate pair $(R_{01},R_{02})$ is said to be \emph{achievable for a target joint distribution} $q_{X_1X_2WY}$ \emph{with cribbing} provided there exists a sequence of $(2^{nR_{01}}, 2^{nR_{02}}, n)$ codes such that
\end{defn} 
\begin{align}\label{eqn:correctnessnoisy}
\lim_{n \to \infty} ||p^{\text{ind}}_{X_1^n,X_2^n,W^n,Y^n}-q^{(n)}_{X_1X_2WY}||_{1}=0,
\end{align}
where $q^{(n)}_{X_1X_2WY}$ is the target i.i.d. product distribution 
\begin{align*}
&q^{(n)}_{X_1X_2WY}(x_1^n,x_2^n,w^n,y^n) :=\prod_{i=1}^n q_{X_1X_2WY}(x_{1i},x_{2i},w_i,y_i).
\end{align*} 

\begin{defn}\label{defn:newnoisy}
The \emph{rate region} $\mathcal{R}_{\textup{noisy-coord}}^{\textup{MAC, crib}}$ is the closure of the set of all achievable rate pairs $(R_{01},R_{02})$. 
\end{defn}

We also separately consider the case when the switch $S$ in Fig.~\ref{fig:encSRnoisy} is open. A code, an achievable rate pair, and the rate region can be defined analogously. In particular, the code and an achievable rate pair can be defined similar to Definitions~\ref{defn:codenoisy} and \ref{def:achnoisy} by changing the map at Encoder~$1$ to simply $p^{\text{Enc}_1}(\tilde{x}_1^n|x_1^n,k_1)$. The \emph{rate region} $\mathcal{R}_{\textup{noisy-coord}}^{\textup{MAC}}$ is the closure of the set of all achievable rate tuples $(R_{01},R_{02})$ when the switch $S$ in Fig.\ref{fig:encSRnoisy} is open. 
Let $\mathcal{R}_{\textup{noisy-coord, $R_{02} \to \infty$}}^{\textup{MAC}}$ be the region when the pairwise shared randomness $K_2$ is unlimited: 
\begin{align}&\mathcal{R}_{\textup{noisy-coord, $R_{02} \to \infty$}}^{\textup{MAC}}=\{R_{01} : \exists \ R_{02} \nonumber\\
&\text{s.t.}\ (R_{01},R_{02}) \in \mathcal{R}_{\textup{noisy-coord}}^{\textup{MAC}}\}.
\end{align}

\section{Main Results}  \label{secton:mainresults1noisy}
We first present our results in the context where one of the encoders is allowed to crib~\cite[Situation 4]{willems1985discrete} from the other encoder's input non-causally (the switch $S$ in Figure~\ref{fig:encSRnoisy} is closed). This will facilitate cooperation between the encoders, in that Enc~$1$ can build its codebooks conditioned on the knowledge of the input codeword from Enc~$2$.
We have the following inner bound to the rate region $\mathcal{R}_{\textup{noisy-coord}}^{\textup{MAC, crib}}$. In the theorem below, the auxiliary random variables $U_{1}$, $U_{2}$ are used in the joint source-channel coding scheme to send source descriptions of $X_{1}$, $X_{2}$ respectively, while $T$ is a time-sharing random variable. The decoder then recovers the source descriptions and locally simulates $Y$.
\begin{theorem}[Achievable Rate Region with Cribbing amongst the Encoders] \label{thm:encsideIBnoisyK0}
Given a target joint p.m.f. $q_{X_1X_2WY}$, the rate pair $(R_{01},R_{02})$ is in $\mathcal{R}_{\textup{noisy-coord}}^{\textup{MAC, crib}}$ provided
\begin{subequations}
\begin{align} 
I(U_{1};U_2,W,\tilde{Y}|T) &\geq I(U_{1};X_1,\tilde X_2|T) \label{eq:4a}\\
I(U_{2};U_1,W,\tilde{Y}|T) &\geq I(U_{2};X_2|T) \label{eq:4b}\\
I(U_{1},U_{2};W,\tilde{Y}|T) &\geq I(U_{1};X_1,\tilde X_2|T)+I(U_{2};X_2|T) \notag\\
&\hspace{12pt}-I(U_{1};U_2|T) \label{eq:4c}\\
R_{01} &\geq I(U_{1};X_1,X_2,W,Y|T) \notag\\
&\hspace{12pt}-I(U_{1};U_{2},W,\tilde{Y}|T) \label{eq:4d}\\
R_{02} &\geq I(U_{2};X_1,X_2,W,Y|T) \notag\\
&\hspace{12pt}-I(U_{2};U_{1},W,\tilde{Y}|T) \label{eq:4e}\\
R_{01} &\geq I(U_{1};X_1,X_2,W,Y|T)\!-\!I(U_{1};W,\tilde{Y}|T) \notag\\
&\hspace{6pt}+I(U_{2};X_2|T)-I(U_{2};U_1,W,\tilde{Y}|T) \label{eq:4f}\\
R_{02} &\geq I(U_{2};X_1,X_2,W,Y|T)\!-\!I(U_{2};W,\tilde{Y}|T) \notag\\
&\hspace{6pt}+I(U_{1};X_1,\tilde X_2|T)-I(U_{1};U_2,W,\tilde{Y}|T) \label{eq:4g}\\
R_{01}+R_{02} &\geq I(U_{1},U_{2};X_1,X_2,W,Y|T) \notag\\
&\hspace{12pt}-I(U_{1},U_{2};W,\tilde{Y}|T) \label{eq:4h},
\end{align}
\end{subequations}
for some p.m.f. 
\begin{align}
p(x_1,&x_2,w,t,u_{1},u_{2},\tilde{x}_1,\tilde{x}_2,\tilde{y},y)=\nonumber\\
&\hspace{12pt}p(x_1,x_2,w)p(t) p(u_{2},\tilde{x}_2|x_2,t)p(u_{1},\tilde{x}_1|x_1,\tilde x_2,t) \nonumber\\
&\hspace{18pt}\times p(\tilde{y}|\tilde{x}_1,\tilde{x}_2)p(y|u_{1},u_{2},w,\tilde{y},t)\label{eq:pmfthm1noisyK0}
\end{align}
such that 
\begin{align*}
&\sum\limits_{u_{1},u_{2},\tilde{x}_1,\tilde{x}_2,\tilde{y}} \!\!\!\!\!\!\!\!\!\!p(x_1,x_2,w,u_{1},u_{2},\tilde{x}_1,\tilde{x}_2,\tilde{y},y|t) =q(x_1,x_2,w,y) \:\: \forall \:\: t.
\end{align*}
\end{theorem}

We note that constraints \eqref{eq:4a}--\eqref{eq:4c} ensure that the source descriptions can be successfully recovered at the decoder, while constraints \eqref{eq:4d}--\eqref{eq:4h} are the minimum rates of shared randomness needed for channel simulation. In particular, note that the right-hand sides of the inequalities \eqref{eq:4a}, \eqref{eq:4c} and \eqref{eq:4g} featuring mutual information terms with $U_1$ can depend upon $\tilde X_2$. For a detailed proof, please refer to Section~\ref{app:pfThm1noisy}.   


Now consider the case when the switch $S$ in Figure~\ref{fig:encSRnoisy} is open, i.e., no cribbing is admissible. Then we can derive a tight characterization if the channel $p(\tilde{y}|\tilde{x}_1,\tilde{x}_2)$ is composed of deterministic links, i.e., $\tilde{Y}=(f_1(\tilde{X}_1),f_2(\tilde{X}_2))$ for deterministic maps $f_1(\cdot)$ and $f_2(\cdot)$, and $I(X_1;X_2|W)=0$. 
\begin{theorem}[Tight Characterization without Encoder Cribbing] \label{thm:indepnoisy}
Consider a target p.m.f. $q_{X_1X_2WY}$ such that $I(X_1;X_2|W)=0$, and also assume that the MAC is composed of deterministic links, i.e. $\tilde{Y}= (f_1(\tilde X_1),f_2(\tilde X_2)) \triangleq (\tilde Y_1,\tilde Y_2)$. Then the rate region $\mathcal{R}_{\textup{noisy-coord, $R_{02} \to \infty$}}^{\textup{MAC}}$ is characterized by the set of rates $R_{01}$ such that
\begin{subequations}
\begin{align} 
H(\tilde{Y}_1|W,T) &\geq I(U_{1},\tilde{Y}_1;X_1|W,T) \label{eq:3a}\\
H(\tilde{Y}_2|W,T) &\geq I(U_{2},\tilde{Y}_2;X_2|W,T) \label{eq:3b}\\
R_{01} &\geq I(U_1,\tilde Y_1;X_1,Y|X_2,W,T)-H(\tilde Y_1|W,T) \label{eq:3c},
\end{align}
\end{subequations}
for some p.m.f. 
\begin{align}
p(x_1,&x_2,w,t,u_{1},u_{2},\tilde{x}_1,\tilde{x}_2,\tilde y,y)=\nonumber\\
&\hspace{0pt}p(w)p(x_1|w)p(x_2|w)p(t) \prod_{j=1}^2 p(u_{j}|x_j,t)p(\tilde{x}_j|u_j,x_j,t)  \nonumber\\
&\hspace{18pt}\times p(\tilde{y}|\tilde{x}_1,\tilde{x}_2)p(y|u_{1},u_{2},w,\tilde{y},t)  \label{pmf:ob3noisy}
\end{align}
such that 
\begin{align*}
&\!\!\sum\limits_{u_{1},u_{2},\tilde{x}_1,\tilde{x}_2,\tilde{y}} \!\!\!\!\!\!\!\!\!\!p(x_1,x_2,w,u_{1},u_{2},\tilde{x}_1,\tilde{x}_2,\tilde{y},y|t) \!=\!q(x_1,x_2,w,y) \: \forall \: t,
\end{align*}
with the auxiliary cardinalities bounded as $|\mathcal{U}_{1}| \leq |\mathcal{X}_1||\mathcal{X}_2||\mathcal{W}||\mathcal{\tilde X}_1||\mathcal{\tilde X}_2||\mathcal{\tilde Y}||\mathcal{Y}|$, $|\mathcal{U}_{2}| \leq |\mathcal{U}_{1}||\mathcal{X}_1||\mathcal{X}_2||\mathcal{W}||\mathcal{\tilde X}_1||\mathcal{\tilde X}_2||\mathcal{\tilde Y}||\mathcal{Y}|$ and $|\mathcal{T}| \leq 3$.
\end{theorem}
The achievability largely follows from~\cite[Theorem 3]{HaddadpourYAG13}, by also accounting for the decoder side information and enforcing the conditional independence $p(x_1,x_2,w) = p(w)p(x_1|w)p(x_2|w)$ along with $\tilde{Y}=(\tilde Y_1,\tilde Y_2)$. A detailed proof of the converse is given in Section~\ref{proof:conv-thm5}.

\section{Example: Cribbing Helps for Channel Simulation} \label{sec:example}
In this section, we show with the help of an example that in the presence of cribbing between the encoders, the achievable region can be improved upon. Our illustration will be in the context of Theorem~\ref{thm:indepnoisy}. Accordingly, we first compute the region of Theorem~\ref{thm:indepnoisy} without encoder cribbing for this example, and then show that the region is improved in the presence of encoder cribbing.  
\begin{example} \label{ex:1}
Let $X_1=(X_{11},X_{12})$, where $X_{11}$ and $X_{12}$ are independent and uniform binary random variables. Let $X_2=B$, where $B$ is uniformly distributed on $\{1,2\}$ and independent of $X_1$. Suppose the channel to be simulated $q_{Y|X_1,X_2}$ is such that $Y=X_{1B}$. For the sake of simplicity, we let $W=\varnothing$ and assume a perfect resource channel $\tilde Y=(\tilde X_1,\tilde X_2)$. We focus on the requisite values of $H(\tilde X_1)$ and $H(\tilde X_2)$ for channel simulation in the presence of unlimited shared randomness rates, with and without encoder cribbing.   
\end{example}

In the absence of cribbing between the encoders, from Theorem~\ref{thm:indepnoisy}, the region $\mathcal{R}_{\textup{noisy-coord}}^{\textup{MAC}}$ for unlimited shared randomness rates $R_{01},R_{02}$ and a perfect channel $\tilde Y_j=\tilde X_j$ for $j \in \{1,2\}$ is simply characterized by the feasibility constraints
\begin{align*} 
H(\tilde{X}_1|T) &\geq I(U_{1},\tilde{X}_1;X_1|T) \\
H(\tilde{X}_2|T) &\geq I(U_{2},\tilde{X}_2;X_2|T),
\end{align*}
for some p.m.f. 
\begin{align*}
p(x_1,&x_2,t,u_{1},u_{2},\tilde{x}_1,\tilde{x}_2,y)=\nonumber\\
&\hspace{6pt}p(x_1)p(x_2)p(t) \prod_{j=1}^2 p(u_{j},\tilde x_j|x_j,t) p(y|u_{1},u_{2},\tilde{x}_1,\tilde x_2,t)  
\end{align*}
such that 
{$\sum\limits_{u_{1},u_{2},\tilde{x}_1,\tilde{x}_2}p(x_1,x_2,u_{1},u_{2},\tilde{x}_1,\tilde{x}_2,y|t)=q(x_1,x_2,y),$}
for all $t$. The following proposition explicitly characterizes the optimal region (of feasibility constraints) for the given $q_{X_1 X_2 Y}$.

\begin{prop} \label{prop:1}
For the target distribution $q_{X_1 X_2 Y}$ in Example~\ref{ex:1}, channel simulation is feasible if and only if $H(\tilde X_1) \geq 2$ and $H(\tilde X_2) \geq 1$ (thus $H(\tilde X_1,\tilde X_2) \geq 3$).
\end{prop}
\begin{proof}
For the achievability, it suffices to prove that there exists a p.m.f.
\begin{align*}
p(x_1,&x_2,u_{1},u_{2},\tilde{x}_1,\tilde{x}_2,y)=\\
&\hspace{6pt}p(x_1)p(x_2)\prod_{j=1}^2 p(u_{j},\tilde x_j|x_j) p(y|u_{1},u_{2},\tilde{x}_1,\tilde x_2)   
\end{align*}
such that $\sum\limits_{u_{1},u_{2},\tilde{x}_1,\tilde{x}_2}p(x_1,x_2,u_{1},u_{2},\tilde{x}_1,\tilde{x}_2,y)=q(x_1,x_2,y)$, $I(U_{1},\tilde{X}_1;X_1)=2$ and $I(U_{2},\tilde{X}_2;X_2)=1$. By choosing $U_1=X_1$ and $U_2=X_2$, it is clear that the conditions on the joint p.m.f. $p(x_1,x_2,u_{1},u_{2},\tilde{x}_1,\tilde{x}_2,y)$ are satisfied and $I(U_{1},\tilde{X}_1;X_1)=H(X_1)=2$, $I(U_{2},\tilde{X}_2;X_2)=H(B)=1$. The interesting part is the converse, in Appendix~\ref{sec:exproof}.
\end{proof}

We next prove that with cribbing, it is possible to 
achieve channel simulation with smaller values of $H(\tilde X_1)$ and $H(\tilde X_2)$ compared to the optimal region without cribbing. To see this, we choose $\tilde X_2 = X_2$ which in the presence of cribbing makes $X_2=B$ available to Enc~$1$. Then Enc~$1$ can afford to send only $\tilde X_1=X_{1B}$ (instead of the entire $(X_{11},X_{12})$ necessitated in the absence of cribbing). More formally, the region $\mathcal{R}_{\textup{noisy-coord}}^{\textup{MAC, crib}}$ in Theorem~\ref{thm:encsideIBnoisyK0} specializes for independent sources, perfect channel and unlimited shared randomness to the set of feasibility constraints
\begin{align*}
H(\tilde{X}_1|T) &\geq I(U_{1},\tilde X_1;X_1,\tilde X_2|T)\!-\!I(U_1;U_{2},\tilde X_2|\tilde X_1,T) \\
H(\tilde{X}_2|T) &\geq I(U_{2},\tilde{X}_2;X_2|T)-I(U_2;U_{1},\tilde X_1|\tilde X_2,T)\\
H(\tilde X_1,\tilde{X}_2|T) &\geq I(U_{2},\tilde{X}_2;X_2|T)-I(U_{2},\tilde{X}_2;U_{1},\tilde{X}_1|T) \\
&\phantom{wwww}+I(U_{1},\tilde X_1;X_1,\tilde X_2|T),
\end{align*}
for some p.m.f. 
\begin{align}
p(x_1,&x_2,t,u_{1},u_{2},\tilde{x}_1,\tilde{x}_2,y)=\nonumber\\
&\hspace{6pt}p(x_1)p(x_2)p(t) p(u_{2},\tilde x_2|x_2,t)p(u_{1},\tilde x_1|x_1,\tilde x_2,t)\nonumber\\
&\hspace{18pt} \times p(y|u_{1},u_{2},\tilde{x}_1,\tilde x_2,t)  
\end{align}
such that 
\begin{align*}
&\sum\limits_{u_{1},u_{2},\tilde{x}_1,\tilde{x}_2}p(x_1,x_2,u_{1},u_{2},\tilde{x}_1,\tilde{x}_2,y|t)=q(x_1,x_2,y),
\end{align*}
for all $t$. Now we can choose $U_1=X_{1B}$ and $U_2=B$ to obtain that channel simulation is feasible if $H(\tilde X_1) \geq 1$, $H(\tilde X_2) \geq 1$ and $H(\tilde X_1,\tilde X_2) \geq 2$. This strictly improves upon the (optimal region of) feasibility constraints without cribbing, which were $H(\tilde X_1) \geq 2$, $H(\tilde X_2) \geq 1$ and $H(\tilde X_1,\tilde X_2) \geq 3$.

\section{Proof of Theorem~\ref{thm:encsideIBnoisyK0}} \label{app:pfThm1noisy}
The proof makes use of the Output Statistics of Random Binning (OSRB) framework~\cite{yassaee2014achievability}. In the following discussion, we adopt the convention of using capital letters (such as $P_{X}$) to represent random p.m.f.'s, as in \cite{cuff2013distributed,yassaee2014achievability}. 
We establish achievability with $|\mathcal{T}|=1$, for simplicity. \\ 
\underline{Random Binning Protocol:} Let the random variables $(U_{1}^n,U_{2}^n,X_1^n,X_2^n,W^n,\tilde{X}_1^n,\tilde{X}_2^n,\tilde{Y}^n,Y^n)$ be drawn i.i.d. according to the joint distribution
\begin{align}
&p(x_1,x_2,w,u_{1},u_{2},\tilde{x}_1,\tilde{x}_2,\tilde{y},y) \notag\\
&=p(x_1,x_2,w)p(u_{2}|x_2)p(\tilde{x}_2|u_{2},x_2)p(u_{1}|x_1,\tilde x_2) \notag\\
&\phantom{ww}\times p(\tilde{x}_1|u_{1},x_1,\tilde x_2)p(\tilde{y}|\tilde{x}_1,\tilde{x}_2)p(y|u_{1},u_{2},w,\tilde{y})
\end{align}
such that the marginal $p(x_1,x_2,w,y)=q(x_1,x_2,w,y)$. 
The following random binning is then applied: independently generate two uniform bin indices $(K_j,F_j)$ of $U_{j}^n$, where $K_j = \phi_{j1}(U_{j}^n) \in [1:2^{n R_{0j}}]$ and $F_j = \phi_{j2}(U_{j}^n) \in [1:2^{n \tilde{R}_j}]$.
The receiver estimates $(\hat{u}_{1}^n,\hat{u}_{2}^n)$ from its observations $(k_1,f_1,k_2,f_2,w^n,\tilde{y}^n)$ using a Slepian-Wolf decoder.
The random p.m.f. induced by this binning scheme is given by:
\begin{align*}
&P(x_1^n,x_2^n,w^n,y^n,u_{1}^n,u_{2}^n,\tilde{x}_1^n,\tilde{x}_2^n,\tilde{y}^n,k_1,f_1,k_2,f_2,\hat{u}_{1}^n,\hat{u}_{2}^n) \notag\\
&= p(x_1^n,x_2^n,w^n)p(u_{2}^n|x_2^n)p(\tilde{x}_2^n|u_{2}^n,x_2^n)p(u_{1}^n|x_1^n,\tilde x_2^n) \notag\\
&\phantom{w} \times p(\tilde{x}_1^n|u_{1}^n,x_1^n,\tilde x_2^n)p(\tilde{y}^n|\tilde{x}_1^n,\tilde{x}_2^n)p(y^n|u_{1}^n,u_{2}^n,w^n,\tilde{y}^n) \notag\\
&\phantom{w} \times P(k_1,f_1|u_{1}^n) P(k_2,f_2|u_{2}^n) \notag\\
&\phantom{w} \times P^{SW}(\hat{u}_{1}^n,\hat{u}_{2}^n|k_1,f_1,k_2,f_2,w^n,\tilde{y}^n)  \\
&= p(x_1^n,x_2^n,w^n) P(k_2,f_2,u_{2}^n|x_2^n) p(\tilde{x}_2^n|u_{2}^n,x_2^n)  \notag\\
&\phantom{w} \times P(k_1,f_1,u_{1}^n|x_1^n,\tilde x_2^n) p(\tilde{x}_1^n|u_{1}^n,x_1^n,\tilde x_2^n)p(\tilde{y}^n|\tilde{x}_1^n,\tilde{x}_2^n) \notag\\
&\phantom{w} \times P^{SW}(\hat{u}_{1}^n,\hat{u}_{2}^n|k_1,f_1,k_2,f_2,w^n,\tilde{y}^n) p(y^n|u_{1}^n,u_{2}^n,w^n,\tilde{y}^n) \notag\\
&= p(x_1^n,x_2^n,w^n) P(k_2,f_2|x_2^n) P(u_{2}^n|k_2,f_2,x_2^n) \notag\\
&\phantom{w} \times p(\tilde{x}_2^n|u_{2}^n,x_2^n) P(k_1,f_1|x_1^n,\tilde x_2^n)  \notag\\
&\phantom{w} \times P(u_{1}^n|k_1,f_1,x_1^n,\tilde x_2^n) p(\tilde{x}_1^n|u_{1}^n,x_1^n,\tilde x_2^n)p(\tilde{y}^n|\tilde{x}_1^n,\tilde{x}_2^n) \notag\\
&\phantom{w} \times P^{SW}(\hat{u}_{1}^n,\hat{u}_{2}^n|k_1,f_1,k_2,f_2,w^n,\tilde{y}^n) p(y^n|u_{1}^n,u_{2}^n,w^n,\tilde{y}^n). 
\end{align*}

\noindent \underline{Random Coding Protocol:} In this scheme, we assume the presence of additional shared randomness $F_j$ of rate $\tilde{R}_j, j \in \{1,2\}$ between the respective encoders and the decoder in the original problem. 
Encoder $2$ observes $(k_2,f_2,x_2^n)$, and generates $u_{2}^n$ according to the p.m.f. $P(u_{2}^n|k_2,f_2,x_2^n)$ from the protocol above. Further, encoder $2$ draws $\tilde{x}_2^n$ according to the p.m.f. $p(\tilde{x}_2^n|u_{2j}^n,x_2^n)$. Encoder $1$ observes $(k_1,f_1,x_1^n,\tilde x_2^n)$, and generates $u_{1}^n$ according to the p.m.f. $P(u_{1}^n|k_1,f_1,x_1^n,\tilde x_2^n)$ from the protocol above. Further, encoder $1$ draws $\tilde{x}_1^n$ according to the p.m.f. $p(\tilde{x}_1^n|u_{1}^n,x_1^n,\tilde x_2^n)$. 
The receiver first estimates $(\hat{u}_{1}^n,\hat{u}_{2}^n)$ from its observations $(k_1,k_2,f_1,f_2,w^n,\tilde{y}^n)$ using the Slepian-Wolf decoder from the binning protocol, i.e. $P^{SW}(\hat{u}_{1}^n,\hat{u}_{2}^n|k_1,f_1,k_2,f_2,w^n,\tilde{y}^n)$. 
Then it generates the output $y^n$ according to the distribution $p_{Y^n|U_{1}^n,U_{2}^n,W^n,\tilde{Y}^n}(y^n|\hat{u}_{1}^n,\hat{u}_{2}^n,w^n,\tilde{y}^n)$. 
The induced random p.m.f. from the random coding scheme is
\begin{align*}
&\hat{P}(x_1^n,x_2^n,w^n,y^n,u_{1}^n,u_{2}^n,\tilde{x}_1^n,\tilde{x}_2^n,\tilde{y}^n,k_1,f_1,k_2,f_2,\hat{u}_{1}^n,\hat{u}_{2}^n) \notag\\
&= p^{\text{U}}(k_1)p^{\text{U}}(f_1)p^{\text{U}}(k_2)p^{\text{U}}(f_2)p(x_1^n,x_2^n,w^n)  \notag\\
&\phantom{ww} \times P(u_{2}^n|k_2,f_2,x_2^n) p(\tilde{x}_2^n|u_{2}^n,x_2^n)  \notag\\
&\phantom{w} \times P(u_{1}^n|k_1,f_1,x_1^n,\tilde x_2^n) p(\tilde{x}_1^n|u_{1}^n,x_1^n,\tilde x_2^n)p(\tilde{y}^n|\tilde{x}_1^n,\tilde{x}_2^n) \notag\\
&\phantom{w} \times P^{SW}(\hat{u}_{1}^n,\hat{u}_{2}^n|k_1,f_1,k_2,f_2,w^n,\tilde{y}^n) p(y^n|\hat{u}_{1}^n,\hat{u}_{2}^n,w^n,\tilde{y}^n). 
\end{align*}


\noindent \underline{Analysis of Rate Constraints:}\\
Using the fact that $(k_j,f_j)$ are bin indices of $u_{j}^n$ for $j \in \{1,2\}$, we impose the conditions
\begin{align}
R_{01}+\tilde{R}_1 &\leq H(U_{1}|X_1,\tilde X_2), \label{eq:cond11noisy} \\
R_{02}+\tilde{R}_2 &\leq H(U_{2}|X_2), \label{eq:cond12noisy}
\end{align}
{that ensure, by invoking \cite[Theorem 1]{yassaee2014achievability}
\begin{align}
&P(x_1^n,x_2^n,w^n,u_{1}^n,u_{2}^n,\tilde{x}_1^n,\tilde{x}_2^n,\tilde{y}^n,k_1,f_1,k_2,f_2,\hat{u}_{1}^n,\hat{u}_{2}^n) \notag\\
&\phantom{w} \approx \hat{P}(x_1^n,x_2^n,w^n,u_{1}^n,u_{2}^n,\tilde{x}_1^n,\tilde{x}_2^n,\tilde{y}^n,k_1,f_1,k_2,f_2,\hat{u}_{1}^n,\hat{u}_{2}^n). \label{eq:condit1noisy}
\end{align}

We next impose the following constraints for the success of the (Slepian-Wolf) decoder by Slepian-Wolf theorem~\cite{slepian1973noiseless}
\begin{align}
R_{01}+\tilde{R}_1 &\geq H(U_{1}|U_{2},W,\tilde{Y}), \label{eq:cond21noisy} \\
R_{02}+\tilde{R}_2 &\geq H(U_{2}|U_{1},W,\tilde{Y}), \label{eq:cond22noisy} \\
R_{01}+\tilde{R}_1+R_{02}+\tilde{R}_2 &\geq H(U_{1},U_{2}|W,\tilde{Y}), \label{eq:cond23noisy}
\end{align}
Expressions \eqref{eq:cond21noisy}--\eqref{eq:cond23noisy} suffice to obtain
\begin{align}
&P(x_1^n,x_2^n,w^n,u_{1}^n,u_{2}^n,\tilde{x}_1^n,\tilde{x}_2^n,\tilde{y}^n,k_1,f_1,k_2,f_2,\hat{u}_{1}^n,\hat{u}_{2}^n) \notag\\
&\phantom{w} \approx P(x_1^n,x_2^n,w^n,u_{1}^n,u_{2}^n,\tilde{x}_1^n,\tilde{x}_2^n,\tilde{y}^n,k_1,f_1,k_2,f_2) \notag\\
&\phantom{w} \times \mathbbm{1}\{\hat{u}_{1}^n=u_{1}^n,\hat{u}_{2}^n=u_{2}^n\}. \label{eq:condit2noisy}
\end{align}
Using \eqref{eq:condit2noisy} and \eqref{eq:condit1noisy} in conjunction with the first and third parts of \cite[Lemma 4]{yassaee2014achievability}, we obtain
\begin{align}
&\hat{P}(x_1^n,x_2^n,w^n,u_{1}^n,u_{2}^n,\tilde{x}_1^n,\tilde{x}_2^n,\tilde{y}^n,k_1,f_1,k_2,f_2,\hat{u}_{1}^n,\hat{u}_{2}^n,y^n) \notag\\
&= \hat{P}(x_1^n,x_2^n,w^n,u_{1}^n,u_{2}^n,\tilde{x}_1^n,\tilde{x}_2^n,\tilde{y}^n,k_1,f_1,k_2,f_2,\hat{u}_{1}^n,\hat{u}_{2}^n) \notag\\
&\phantom{w} \times p(y^n|\hat{u}_{1}^n,\hat{u}_{2}^n,w^n,\tilde{y}^n) \notag\\
&\approx P(x_1^n,x_2^n,w^n,u_{1}^n,u_{2}^n,\tilde{x}_1^n,\tilde{x}_2^n,\tilde{y}^n,k_1,f_1,k_2,f_2) \notag\\
&\phantom{w} \times \mathbbm{1}\{\hat{u}_{1}^n=u_{1}^n,\hat{u}_{2}^n=u_{2}^n\} p(y^n|\hat{u}_{1}^n,\hat{u}_{2}^n,w^n,\tilde{y}^n) \notag\\
&= P(x_1^n,x_2^n,w^n,u_{1}^n,u_{2}^n,\tilde{x}_1^n,\tilde{x}_2^n,\tilde{y}^n,k_1,f_1,k_2,f_2,y^n) \notag\\
&\phantom{w} \times \mathbbm{1}\{\hat{u}_{1}^n=u_{1}^n,\hat{u}_{2}^n=u_{2}^n\}.
\end{align}
This implies, by the first part of \cite[Lemma 4]{yassaee2014achievability}
\begin{align}
\hat{P}(x_1^n,x_2^n,w^n,y^n,f_1,f_2) \approx P(x_1^n,x_2^n,w^n,y^n,f_1,f_2). \label{eq:osrbcond1noisy}
\end{align}

We further require $(X_1^n,X_2^n,W^n,Y^n)$ to be nearly independent of $(F_1,F_2)$, so that the latter can be eliminated. 
This is realized by imposing the following conditions:
\begin{align}
\tilde{R}_1 &\leq H(U_{1}|X_1,X_2,W,Y), \label{eq:cond31secnoisy} \\
\tilde{R}_2 &\leq H(U_{2}|X_1,X_2,W,Y), \label{eq:cond32secnoisy} \\
\tilde{R}_1+\tilde{R}_2 &\leq H(U_{1},U_{2}|X_1,X_2,W,Y). \label{eq:cond33secnoisy}
\end{align}
By \cite[Theorem 1]{yassaee2014achievability}, this suffices to obtain
\begin{align*}
&P(x_1^n,x_2^n,w^n,y^n,f_1,f_2) \approx p^{\text{U}}(f_1) p^{\text{U}}(f_2) p(x_1^n,x_2^n,w^n,y^n), 
\end{align*}
which implies that
\begin{align}
&\hat{P}(x_1^n,x_2^n,w^n,y^n,f_1,f_2) \approx p^{\text{U}}(f_1) p^{\text{U}}(f_2) p(x_1^n,x_2^n,w^n,y^n), \label{eq:rcpmf}
\end{align}
by invoking \eqref{eq:osrbcond1noisy} and the triangle inequality.
Hence there exists a fixed binning with corresponding pmf $\tilde{p}$ such that if we replace $P$ by $\tilde{p}$ in \eqref{eq:rcpmf} and denote the resulting pmf by $\hat{p}$, 
\begin{align*}
\hat{p}(x_1^n,x_2^n,w^n,y^n,f_1,f_2) \approx p^{\text{U}}(f_1) p^{\text{U}}(f_2) p(x_1^n,x_2^n,w^n,y^n).
\end{align*}
Now the second part of \cite[Lemma 4]{yassaee2014achievability} allows us to conclude that there exist instances $F_1 = f_1^{*}, F_2=f_2^{*}$ such that 
\begin{align}
\hat{p}(x_1^n,x_2^n,w^n,y^n|f_1^{*},f_2^{*}) \approx p(x_1^n,x_2^n,w^n,y^n).
\end{align}
Finally on eliminating $(\tilde{R}_1,\tilde{R}_2)$ from equations \eqref{eq:cond11noisy} -- \eqref{eq:cond12noisy}, \eqref{eq:cond21noisy} -- \eqref{eq:cond23noisy} and \eqref{eq:cond31secnoisy} -- \eqref{eq:cond33secnoisy} by the FME procedure, we obtain the rate constraints in Theorem~\ref{thm:encsideIBnoisyK0}.

\section{Converse Proof of Theorem \ref{thm:indepnoisy}}\label{proof:conv-thm5}
Consider a coding scheme that induces a joint distribution on $(X_1^n,X_2^n,W^n,Y^n)$ which satisfies the constraint
\begin{gather}
\lVert p_{X_1^n,X_2^n,W^n,Y^n}-q^{(n)}_{X_1X_2WY}\rVert_1 \leq \epsilon. \label{eqn:totalSR1noisy}
\end{gather}
\begin{lemma} \cite[Lemma 6]{CerviaLLB20} \label{lem:c1noisy}
Let $p_{S^n}$ be such that $||p_{S^n}-q_S^{(n)}||_1 \leq \epsilon$, where $q^{(n)}_S(s^n)=\prod_{i=1}^n q_S(s_i)$, then
\begin{align}
\sum_{i=1}^n I_p(S_i;S_{\sim i}) \leq ng(\epsilon),
\end{align}
where ${ g(\epsilon)=2\sqrt{\epsilon}\left(H(S)\!+\!\log{|\mathcal{S}|\!+\!\log{\frac{1}{\sqrt{\epsilon}}}}\right)}\to 0$ as $\epsilon \to 0$.
\end{lemma}


Let us now prove the first inequality in Theorem~\ref{thm:indepnoisy}. Recall that $\tilde Y= (\tilde Y_1,\tilde Y_2)$. 
\begin{align}
0 &\leq H(\tilde Y_1^n|W^n)-I(\tilde Y_1^n; K_1,X_1^n|W^n) \notag\\
&\leq H(\tilde Y_1^n|W^n)-I(\tilde Y_1^n; X_1^n|K_1,W^n) \notag\\
&\leq \sum_{i=1}^n H(\tilde Y_{1i}|W_i)-\sum_{i=1}^n I(\tilde Y_1^n; X_{1i}|K_1,X_{1,i+1}^n,W_{\sim i},W_i) \notag\\
&\stackrel{(a)}= \sum_{i=1}^n H(\tilde Y_{1i}|W_i)-\sum_{i=1}^n I(K_1,\tilde Y_1^n,X_{1,i+1}^n,W_{\sim i}; X_{1i}|W_i) \notag\\
&\leq \sum_{i=1}^n H(\tilde Y_{1i}|W_i)-\sum_{i=1}^n I(K_1,\tilde Y_{1\sim i},W_{\sim i},\tilde Y_{1i}; X_{1i}|W_i) \notag\\
&\stackrel{(b)}= \sum_{i=1}^n H(\tilde Y_{1i}|W_i)-\sum_{i=1}^n I(U_{1i},\tilde Y_{1i}; X_{1i}|W_i) \notag\\
&= nH(\tilde Y_{1T}|W_T,T)-nI(U_{1T},\tilde Y_{1T}; X_{1T}|W_T,T) \notag\\
&= nH(\tilde Y_{1}|W,T)-nI(U_{1},\tilde Y_{1}; X_{1}|W,T), \label{refnoisy1}
\end{align}
where (a) follows from the joint i.i.d. nature of $(X_{1i},W_i),i=1,\dots,n$ and the independence of $K_1$ from $(X_1^n,W^n)$, while (b) follows from an auxiliary random variable identification $U_{1i}=(K_1,\tilde Y_{1\sim i},W_{\sim i})$. The inequality $H(\tilde{Y}_{2}|W,T) \geq I(U_{2},\tilde Y_2;X_2|W,T)$ follows analogously, with an auxiliary random variable choice given by $U_{2i}=(K_2,\tilde Y_{2\sim i})$. 

The constraint on the shared randomness rate $R_{01}$ is proved next. Consider the sequence of inequalities below:
\begin{align}
&nR_{01} = H(K_1) \geq H(K_1|W^n) \notag\\
&= H(K_1,\tilde Y_1^n|W^n)-H(\tilde Y_1^n|K_1,W^n) \notag\\
&\geq H(K_1,\tilde Y_1^n|X_2^n,W^n)-H(\tilde Y_1^n|W^n) \notag\\
&\geq I(K_1,\tilde Y_1^n;X_1^n,Y^n|X_2^n,W^n)-H(\tilde Y_1^n|W^n) \notag\\
&\geq \sum_{i=1}^n I(K_1,\tilde Y_1^n;X_{1i},Y_i|X_{1,i+1}^n,Y_{i+1}^n,X_2^n,W^n)\!-\!\!\sum_{i=1}^n H(\tilde Y_{1i}|W_i) \notag\\
&= \sum_{i=1}^n I(K_1,\tilde Y_1^n,X_{1,i+1}^n,Y_{i+1}^n,X_{2 \sim i},W_{\sim i};X_{1i},Y_i|X_{2i},W_i) \notag\\
&\phantom{www}-\sum_{i=1}^n I(X_{1,i+1}^n,Y_{i+1}^n,X_{2 \sim i},W_{\sim i};X_{1i},Y_i|X_{2i},W_i) \notag\\
&\phantom{www}-\sum_{i=1}^n H(\tilde Y_{1i}|W_i) \notag\\
&\stackrel{(a)}\geq \sum_{i=1}^n I(K_1,\tilde Y_{1\sim i},W_{\sim i},\tilde Y_{1i};X_{1i},Y_i|X_{2i},W_i)-ng(\epsilon) \notag\\
&\phantom{www}-\sum_{i=1}^n H(\tilde Y_{1i}|W_i) \notag\\
&= \sum_{i=1}^n I(U_{1i},\tilde Y_{1i};X_{1i},Y_i|X_{2i},W_i)-\sum_{i=1}^n H(\tilde Y_{1i}|W_i)-ng(\epsilon) \notag\\
&= nI(U_{1T},\tilde Y_{1T};X_{1T},Y_T|X_{2T},W_T,T)\!-\!nH(\tilde Y_{1T}|W_T,T)\!-\!ng(\epsilon) \notag\\
&= nI(U_{1},\tilde Y_{1};X_{1},Y|X_{2},W,T) -nH(\tilde Y_{1}|W,T)-ng(\epsilon),
\end{align}
where (a) follows 
by \eqref{eqn:totalSR1noisy} and Lemma~\ref{lem:c1noisy}.

The proof of Theorem~\ref{thm:indepnoisy} is completed by showing the continuity of the derived converse bound at $\epsilon=0$ (note that $g(0):=0$ through continuous extension of the function $g(\epsilon)$). This continuity follows from cardinality bounds on the auxiliary random variables $(U_1,U_2)$ to ensure the compactness of the simplex, as outlined in \cite[Lemma VI.5]{cuff2013distributed}. 
The cardinalities of $U_{1}$ and $U_{2}$ can be restricted as mentioned following the perturbation method of \cite{gohari2012evaluation}, similar to \cite{RamachandranOSISIT2024,ramachandran2017feedback}.
Finally, by invoking the continuity properties of total variation distance and mutual information in the probability simplex, the converse for Theorem~\ref{thm:indepnoisy} is complete. 

\balance

\bibliographystyle{IEEEtran}
\bibliography{mylit}
\clearpage

\appendices
\section{Converse Proof of Proposition~\ref{prop:1}}\label{sec:exproof}
For the converse, we need to prove that for any p.m.f. 
\begin{align*}
p(x_1,&x_2,u_{1},u_{2},\tilde{x}_1,\tilde{x}_2,y)=\\
&\hspace{3pt} p(x_1)p(x_2) p(u_{1},\tilde x_1|x_1) p(u_{2},\tilde x_2|x_2) p(y|u_{1},u_{2},\tilde{x}_1,\tilde x_2)
\end{align*}
such that $$\sum\limits_{u_{1},u_{2},\tilde{x}_1,\tilde{x}_2}p(x_1,x_2,u_{1},u_{2},\tilde{x}_1,\tilde{x}_2,y)=q(x_1,x_2,y),$$ we have $I(U_1,\tilde X_1;X_1) \geq 2$ and $I(U_2,\tilde X_2;X_2) \geq 1$. The independence between $X_1$ and $X_2$ along with the long Markov chain $(U_1,\tilde X_1) \to X_1 \to X_2 \to (U_2,\tilde X_2)$ ensures that $(U_1,\tilde X_1,X_1)$ is independent of $(U_2,\tilde X_2,X_2)$. We also have the output Markov chain $Y\rightarrow (U_1,\tilde X_1,U_2,\tilde X_2)\rightarrow (X_1,X_2)$. Clearly, if $H(X_1|U_1,\tilde X_1)=H(X_2|U_2,\tilde X_2)=0$, it follows that $I(U_1,\tilde X_1;X_1)=H(X_1)=2$ and $I(U_2,\tilde X_2;X_2)=H(X_2)=1$. We now prove that if either $H(X_1|U_1,\tilde X_1)>0$ or $H(X_2|U_2,\tilde X_2)>0$, a contradiction arises.

Let $H(X_2|U_2,\tilde X_2)>0$ (which means $I(U_2,\tilde X_2;X_2)=H(X_2)-H(X_2|U_2,\tilde X_2)=1-H(X_2|U_2,\tilde X_2)<1$) for the sake of contradiction. 
Hence there exist $(u_2,\tilde x_2)$ with $P(U_2=u_2,\tilde X_2=\tilde x_2)>0$ such that $p_{X_2|U_2=u_2,\tilde X_2=\tilde x_2}$ is supported on \{1,2\}. We note that the Markov chain $Y\rightarrow (X_1,U_2,\tilde X_2)\rightarrow X_2$ holds because
\begin{align}
&I(Y;X_2|X_1,U_2,\tilde X_2) \leq I(Y,U_1,\tilde X_1;X_2|X_1,U_2,\tilde X_2)\notag\\
&=\!I(U_1,\tilde X_1;X_2|X_1,U_2,\tilde X_2)+I(Y;X_2|U_1,\tilde X_1,U_2,\tilde X_2,X_1)\notag\\
&\leq\!  I(U_1,\tilde X_1,X_1;U_2,\tilde X_2,X_2)\!+I(Y;X_1,X_2|U_1,\tilde X_1,U_2,\tilde X_2)\notag\\
&=0, \label{eq:MCprop}
\end{align}
where the last equality follows because $(U_1,\tilde X_1,X_1)$ is independent of $(U_2,\tilde X_2,X_2)$ and the Markov chain $Y\rightarrow (U_1,\tilde X_1,U_2,\tilde X_2)\rightarrow (X_1,X_2)$ holds.

Let us consider the induced distribution given by $p_{Y|X_1=(0,1),U_2=u_2,\tilde X_2=\tilde x_2}$. This is well-defined because $P(X_1=(0,1),U_2=u_2,\tilde X_2=\tilde x_2)>0$ as $X_1$ is independent of $(U_2,\tilde X_2)$ and $P(X_1=(0,1))>0,P(U_2=u_2,\tilde X_2=\tilde x_2)>0$.  The fact that 
\begin{align*} 
&P(X_2=2|X_1=(0,1),U_2=u_2,\tilde X_2=\tilde x_2)\\
&\phantom{www}=P(X_2=2|U_2=u_2,\tilde X_2=\tilde x_2)>0
\end{align*}
along with the Markov chain $Y\rightarrow (X_1,U_2,\tilde X_2)\rightarrow X_2$ imply  
\begin{align}
&P(Y=0|X_1=(0,1),U_2=u_2,\tilde X_2=\tilde x_2) \notag\\
&=P(Y=0|X_1=(0,1),U_2=u_2,\tilde X_2=\tilde x_2,X_2=2)=0. \label{eq:contra1}
\end{align}
Likewise, the fact that 
\begin{align*}
&P(X_2=1|X_1=(0,1),U_2=u_2,\tilde X_2=\tilde x_2)\\
&\phantom{www}=P(X_2=1|U_2=u_2,\tilde X_2=\tilde x_2)>0
\end{align*}
along with the Markov chain $Y\rightarrow (X_1,U_2,\tilde X_2)\rightarrow X_2$ imply  
\begin{align}
&P(Y=1|X_1=(0,1),U_2=u_2,\tilde X_2=\tilde x_2) \notag\\
&=P(Y=1|X_1=(0,1),U_2=u_2,\tilde X_2=\tilde x_2,X_2=1)=0. \label{eq:contra2} 
\end{align}
From \eqref{eq:contra1} and \eqref{eq:contra2}, we are led to a contradiction since $p_{Y|X_1=(0,1),U_2=u_2,\tilde X_2=\tilde x_2}$ has to be a probability distribution.

Similarly, 
let $H(X_1|U_1,\tilde X_1)>0$ (which means $I(U_1,\tilde X_1;X_1)=H(X_1)-H(X_1|U_1,\tilde X_1)=2-H(X_1|U_1,\tilde X_1)<2$) for the sake of contradiction. 
Hence there exist $(u_1,\tilde x_1)$ with $P(U_1=u_1,\tilde X_1=\tilde x_1)>0$ such that $p_{X_1|U_1=u_1,\tilde X_1=\tilde x_1}$ has a support whose size is larger than $1$. This means that the support can be a superset of $\{(0,0),(0,1)\}$, $\{(0,0),(1,0)\}$, $\{(0,0),(1,1)\}$, $\{(0,1),(1,0)\}$, $\{(0,1),(1,1)\}$ or $\{(1,0),(1,1)\}$ -- these are considered in turn next. We note that the Markov chain $Y\rightarrow (X_2,U_1,\tilde X_1)\rightarrow X_1$ holds via similar reasoning as \eqref{eq:MCprop}.

Suppose that the support of $p_{X_1|U_1=u_1,\tilde X_1=\tilde x_1}$ is a superset of $\{(0,0),(0,1)\}$. Let us consider the induced distribution given by $p_{Y|X_2=2,U_1=u_1,\tilde X_1=\tilde x_1}$. This is well-defined because $P(X_2=2,U_1=u_1,\tilde X_1=\tilde x_1)>0$ as $X_2$ is independent of $(U_1,\tilde X_1)$ and $P(X_2=2)>0,P(U_1=u_1,\tilde X_1=\tilde x_1)>0$.  The fact that $P(X_1=(0,1)|X_2=2,U_1=u_1,\tilde X_1=\tilde x_1)>0$ along with the Markov chain $Y\rightarrow (X_2,U_1,\tilde X_1)\rightarrow X_1$ imply  
\begin{align*}
&P(Y=0|X_2=2,U_1=u_1,\tilde X_1=\tilde x_1) \notag\\
&=P(Y=0|X_2=2,U_1=u_1,\tilde X_1=\tilde x_1,X_1=(0,1))=0. 
\end{align*}
Likewise, the fact that $P(X_1=(0,0)|X_2=2,U_1=u_1,\tilde X_1=\tilde x_1)>0$ along with the Markov chain $Y\rightarrow (X_2,U_1,\tilde X_1)\rightarrow X_1$ imply  
\begin{align*}
&P(Y=1|X_2=2,U_1=u_1,\tilde X_1=\tilde x_1) \notag\\
&=P(Y=1|X_2=2,U_1=u_1,\tilde X_1=\tilde x_1,X_1=(0,0))=0. 
\end{align*}
We are led to a contradiction since $p_{Y|X_2=2,U_1=u_1,\tilde X_1=\tilde x_1}$ has to be a probability distribution.

Next suppose that the support of $p_{X_1|U_1=u_1,\tilde X_1=\tilde x_1}$ is a superset of $\{(0,0),(1,0)\}$. We consider the well-defined distribution $p_{Y|X_2=1,U_1=u_1,\tilde X_1=\tilde x_1}$. The fact that $P(X_1=(1,0)|X_2=1,U_1=u_1,\tilde X_1=\tilde x_1)>0$ along with the Markov chain $Y\rightarrow (X_2,U_1,\tilde X_1)\rightarrow X_1$ imply  
\begin{align*}
&P(Y=0|X_2=1,U_1=u_1,\tilde X_1=\tilde x_1) \notag\\
&=P(Y=0|X_2=1,U_1=u_1,\tilde X_1=\tilde x_1,X_1=(1,0))=0. 
\end{align*}
Likewise, the fact that $P(X_1=(0,0)|X_2=1,U_1=u_1,\tilde X_1=\tilde x_1)>0$ along with the Markov chain $Y\rightarrow (X_2,U_1,\tilde X_1)\rightarrow X_1$ imply  
\begin{align*}
&P(Y=1|X_2=1,U_1=u_1,\tilde X_1=\tilde x_1) \notag\\
&=P(Y=1|X_2=1,U_1=u_1,\tilde X_1=\tilde x_1,X_1=(0,0))=0. 
\end{align*}
We are led to a contradiction since $p_{Y|X_2=1,U_1=u_1,\tilde X_1=\tilde x_1}$ has to be a probability distribution.
The other supports $\{(0,0),(1,1)\}$, $\{(0,1),(1,0)\}$, $\{(0,1),(1,1)\}$ and $\{(1,0),(1,1)\}$ can be analyzed in a similar manner to arrive at a contradiction.

\end{document}